\documentclass[nonacm,sigconf]{acmart}
\usepackage{dsfont}
\usepackage{graphicx}
\usepackage[ruled, linesnumbered]{algorithm2e}
\usepackage{pgfplots}
\usepackage{booktabs}
\usepackage[utf8]{inputenc}

\renewcommand\footnotetextcopyrightpermission[1]{}
\setcopyright{none}
\newcommand\scalemath[2]{\scalebox{#1}{\mbox{\ensuremath{\displaystyle #2}}}}
\AtBeginDocument{%
  \providecommand\BibTeX{{%
    \normalfont B\kern-0.5em{\scshape i\kern-0.25em b}\kern-0.8em\TeX}}}

\begin{document}

\title{Estimating the Percolation Centrality of Large Networks through Pseudo-dimension Theory}


\author{Alane Marie de Lima}
\affiliation{%
  \institution{Federal University of Paran\'{a}}
   \city{Curitiba}
  \country{Brazil}}
\email{amlima@inf.ufpr.br}

\author{Murilo V. G. da Silva}
\affiliation{%
  \institution{Federal University of Paran\'{a}}
   \city{Curitiba}
  \country{Brazil}}
\email{murilo@inf.ufpr.br}

\author{Andr\'{e} Lu\'{i}s Vignatti}
\affiliation{%
  \institution{Federal University of Paran\'{a}}
   \city{Curitiba}
  \country{Brazil}}
\email{vignatti@inf.ufpr.br}


\begin{abstract}
In this work we investigate the problem of estimating the percolation centrality of every vertex in a graph. This centrality measure quantifies the importance of each vertex in a graph going through a contagious process. 
It is an open problem whether the percolation centrality can be computed in $\mathcal{O}(n^{3-c})$ time, for any constant $c>0$. In this paper we present a $\mathcal{O}(m \log^2 n)$ randomized approximation algorithm for the percolation centrality for every vertex of $G$, generalizing techniques developed by 
Riondato, Upfal e Kornaropoulos
(this complexity is reduced to $\mathcal{O}((m+n) \log n)$ for unweighted graphs). The estimation obtained by the algorithm is within $\epsilon$ of the exact value with probability $1- \delta$, for {\it fixed} constants $0 < \epsilon,\delta \leq 1$. In fact, we show in our experimental analysis that in the case of real world complex networks, the output produced by our algorithm is significantly closer to the exact values than its guarantee in terms of theoretical worst case analysis.  

\end{abstract}



\maketitle

\section{Introduction}
\label{sec:intro}

The importance of a vertex in a graph can be quantified using centrality measures. In this paper we deal with the \emph{percolation centrality}, a measure relevant in applications where graphs are used to model a contagious process in a network (e.g., disease transmission or misinformation spreading). Centrality measures can be defined in terms of local properties, such as the vertex degree, or global properties, such as the betweenness centrality or the percolation centrality. The betweenness centrality of a vertex $v$,  roughly speaking, is the fraction of shortest paths containing $v$ as an intermediate vertex. The percolation centrality generalizes the betweenness centrality by allowing weights on the shortest paths, and the weight of a shortest path depends on the disparity between the degree of contamination of the two end vertices of such path.


The study of the percolation phenomenon in a physical system was introduced by \cite{Broadbent1957} in the context of the passage of a fluid in a medium.
In graphs, percolation centrality was proposed by Piraveenan {\it et al} (2013) \cite{Piraveenan2013}, where the medium are the vertices of a graph $G$ and each vertex $v$ in $G$ has a {\it percolation state} (reflecting the ``degree of contamination'' of $v$).
The percolation centrality of $v$ is a function that depends on the topological connectivity and the states of the vertices of $G$ (the appropriate formal definitions are given in Section \ref{sec:preliminaries}).

The best known algorithms that exactly compute the betweenness centrality for every vertex of a graph depends on computing all its shortest paths \cite{Riondato2016} and, consequently, the same applies to the computation of percolation centrality. The best known algorithm for this task for weighted graphs runs in time $\mathcal{O}\left(n^3 / 2^{c \sqrt{\log n}}\right)$, for some constant $c$ \cite{ryan}. Currently it is a central open problem in graph theory whether this problem can be solved in $\mathcal{O}(n^{3-c})$, for any $c > 0$ and the hypothesis that there is no such algorithm is used in hardness arguments in some works\cite{abboudwilliams,abboudwilliams2}. In the particular case of sparse graphs, which are common in applications, the complexity of the exact computation for the betweenness centrality can be improved to $\mathcal{O}(n^2)$. However, the same is not known to be the true for percolation centrality and no subcubic algorithm is known even in such restricted scenario.


The present paper uses techniques developed in the work of Riondato and Kornaropoulos (2016) \cite{Riondato2016} and Riondato and Upfal (2018) \cite{RiondatoUpfal} on   betweenness centrality. A main theme in their work is the fact that for large scale graphs, even algorithms with time complexity that scales quadratically are inefficient in practice and high-quality approximations obtained with high confidence are usually sufficient in real world applications. The authors observe that keeping track of the exact centrality values, which may change continuously, provides little information gain. So the idea is to sample a subset of all shortest paths in the graph so that, for given $0 < \epsilon,\delta \leq 1$, they obtain values within $\epsilon$ from the exact value with probability $1 - \delta$
\cite{Riondato2016}. 

We develop the main results of this paper having in mind both a theoretical and a practical perspective. From the theoretical perspective, for $0 < \epsilon,\delta \leq 1$ being any fixed constants, we show that the estimation of the percolation centrality can be done very fast. More precisely, for sparse graphs, the time complexity is $\mathcal{O}(n \log^2 n)$, the same complexity for the estimation of the betweenness centrality in such graphs, even though the exact computation of the first problem is much more expensive than the second problem. In the practical front, in Section \ref{subsec:pseudodimension} we give the relation between these constants and the sample size required for meeting the approximation guarantee and, in fact, our experimental evaluation shows that our algorithm produces results that are orders of magnitude better than the guarantees given by the referred theoretical analysis.

The techniques developed by \cite{Riondato2016} for the betweenness problem relies on the Vapnik-Chervonenkis (VC) dimension theory and the $\epsilon$-sample theorem. In our work, we use such techniques together with pseudo-dimension (a generalization of the VC-dimension) theory to show that the more general problem of estimating the percolation centrality of every vertex of $G$ can be computed in $\mathcal{O}(m \log^2 n)$ time
(this complexity is reduced to $\mathcal{O}((m+n) \log n)$ for unweighted graphs). We note that in the more recent work of Riondato and Upfal \cite{RiondatoUpfal} they also use pseudo-dimension theory for the betweenness problem, but they obtain different bounds for the sample size and they use pseudo-dimension in order to make use of Rademacher Averages. In our work we need pseudo-dimension theory by the very nature of the problem since percolation functions are real-valued and VC-dimension does not apply in our scenario.



\section{Preliminaries}
\label{sec:preliminaries}

We now introduce the definitions, notation and results we use as the groundwork of our proposed algorithms. 

\subsection{Graphs and Percolation Centrality} \label{subsec:graphs}

Given a graph $G = (V,E)$ (directed or undirected), the percolation states $x_v$ for each $v \in V$ and $(u,w) \in V^2$, let $S_{uw}$ be the set of all shortest paths from $u$ to $w$, and $\sigma_{uw} = |S_{uw}|$. For a given path $p_{uw} \in S_{uw}$, let $Int(p_{uw})$ be the set of internal vertices of $p_{uw}$, that is, $Int(p_{uw}) = \{v \in V : v \in p_{uw} \text{ and } u \neq v \neq w \}$. We denote $\sigma_{uw}(v)$ as the number of shortest paths from $u$ to $w$ that $v \in V$ is internal to. Let $P_u(w) = \{s \in V : (s,w) \in E_{p_{uw}}\}$ be the set of (immediate) \emph{predecessors} of $w$ in $p_{uw} \in S_{uw}$, where $E_{p_{uw}}$ is the set of edges of $p_{uw}$. We call the \emph{diameter} of $G$ as the largest shortest path in $G$. Let $0\le x_v \le 1$ be the percolation state of $v \in V$. We say $v$ is \emph{fully percolated} if $x_v = 1$, \emph{non-percolated} if $x_v = 0$ and \emph{partially percolated} if $0 < x < 1$. We say that a path from $u$ to $w$ is \emph{percolated} if $x_u - x_w > 0$. The percolation centrality is defined below. 



\begin{definition}[Percolation Centrality]\label{def:percolation2}
    Let $R(x)= \max\{x,0\}$. Given a graph $G = (V,E)$ and percolation states $x_v, \forall v \in V$, the \emph{percolation centrality} of a vertex $v \in V$ is defined as 
     \[p(v) = \frac{1}{n(n-1)} \sum\limits_{\substack{(u,w) \in V^2 \\ u \neq v \neq w}} \frac{\sigma_{uw}(v)}{\sigma_{uw}} \frac{R(x_u - x_w)}{\sum\limits_{\substack{(f,d) \in V^2 \\ f \neq v \neq d}} R(x_f - x_d)}.\]
\end{definition}

The definition originally presented in \cite{Piraveenan2013} does not have the normalization factor $\frac{1}{n(n-1)}$, introduced here to allow us to define a proper probability distribution in Section \ref{sec:pseudo}. This normalization obviously does not change the relation between the centrality of vertices. 

%
\subsection{Sample Complexity and Pseudo-dimension} \label{subsec:pseudodimension}

In sampling algorithms, the sample complexity analysis relates the minimum size of a random sample required to estimate results that are consistent with the desired parameters of quality and confidence (e.g., in our case a minimum number of shortest paths that must be sampled). An upper bound to the Vapnik-Chervonenkis Dimension (VC-dimension) of a class of binary functions especially defined in order to model the particular problem that one is dealing provides an upper bound to sample size respecting such parameters. Generally speaking, the VC-dimension measures the expressiveness of a class of subsets defined on a set of points \cite{Riondato2016}. 

For the problem presented in this work, however, the class of functions that we need to deal are not binary. Hence, we use the \emph{pseudo-dimension}, which is a generalization of the VC-dimension for real-valued functions. An in-depth exposition of the definitions and results presented below can be found in the books of Shalev-Schwartz and Ben-David \cite{ShalevShwartz2014}, Anthony and Bartlett \cite{Anthony2009}, and Mohri et. al. \cite{Mohri2012}.  

\paragraph*{Empirical averages and $\epsilon$-representative samples} 
Given a domain $U$ and a set $\mathcal{H}$, let $\mathcal{F}$ be the family of functions from $U$ to $\mathbb{R}_+$ such that there is one $f_h \in \mathcal{F}$ for each $h \in \mathcal{H}$. Let $S$ be a collection of $r$ elements from $U$ sampled with respect to a probability distribution $\pi$. 

\begin{definition}\label{def:averages}
    For each $f_h \in \mathcal{F}$, such that $h \in \mathcal{H}$, we define the expectation of $f_h$ and its empirical average as $L_U$ and $L_S$, respectively, i.e.,
\[L_U(f_h) = \mathbb{E}_{u \in U} [f_h(u)] \quad \text{and} \quad L_S(f_h) = \frac{1}{r} \sum_{s \in S} f_h(s).\]
\end{definition}

\begin{definition}\label{def:erepresentative}
    Given $0 < \epsilon,\delta \leq 1$, a set $S$ is called \emph{$\epsilon$-representative} w.r.t. some domain $U$, a set $\mathcal{H}$, a family of functions $\mathcal{F}$ and a probability distribution $\pi$ if
    \[\forall f_h \in \mathcal{F}, |L_S(f_h) - L_U(f_h)| \leq \epsilon.\]
\end{definition}

By the linearity of expectation, the expected value of the empirical average $L_S(f_h)$ corresponds to $L_U(f_h)$. 
Hence, $|L_S(f_h) - L_U(f_h)| = |L_S(f_h) - \mathbb{E}_{f_h \in \mathcal{F}}[L_S(f_h)]|$, and by the \emph{law of large numbers}, $L_S(f_h)$ converges to its true expectation as $r$ goes to infinity, since $L_S(f_h)$ is the empirical average of $r$ random variables sampled independently and identically w.r.t. $\pi$. However, this law provides no information about the value $|L_S(f_h) - L_U(f_h)|$ for any sample size. 
Thus, we use results from the VC-dimension and pseudo-dimension theory, which provide bounds on the size of the sample that guarantees that the maximum deviation of $|L_S(f_h) - L_U(f_h)|$ is within $\epsilon$ with probability at least $1-\delta$, for given $0 < \epsilon,\delta \leq 1$. 

\paragraph*{VC-dimension} A \emph{range space} is a pair $\mathcal{R} = (X, \mathcal{I})$, where $X$ is a domain (finite or infinite) and $\mathcal{I}$ is a collection of subsets of $X$, called \emph{ranges}. For a given $S \subseteq X$, the \emph{projection} of $\mathcal{I}$ on $S$ is the set $\mathcal{I}_S = \{S \cap I : I \in \mathcal{I}\}$. If $|\mathcal{I}_S| = 2^{|S|}$ then we say $S$ is \emph{shattered} by $\mathcal{I}$. The VC-dimension of a range space is the size of the largest subset $S$ that can be shattered by $\mathcal{I}$, i.e., 

\begin{definition}\label{def:vcdim}
    The VC-dimension of a range space $\mathcal{R} = (X,\mathcal{I})$, denoted by $VCDim(\mathcal{R})$, is 
    $VCDim(\mathcal{R}) = \max\{d : \exists S \subseteq X$ such that  $|S| = d \text{ and } |\mathcal{I}_S| = 2^d\}$.
\end{definition}


\paragraph*{Pseudo-dimension} Let $\mathcal{F}$ be a family of functions from some domain $U$ to the range $[0,1]$. Consider $D = U \times [0,1]$. For each $f \in \mathcal{F}$, there is a subset $R_f \subseteq D$ defined as $R_f = \{(x,t) : x \in U \text{ and } t \leq f(x)\}$. 

\begin{definition} [see \cite{Anthony2009}, Section 11.2] \label{def:pseudo}
    Let $\mathcal{R} = (U,\mathcal{F})$ and $\mathcal{R}' = (D,\mathcal{F}^+)$ be range spaces, where $\mathcal{F}^+ = \{R_f : f \in \mathcal{F}\}$. The \emph{pseudo-dimension} of $\mathcal{R}$, denoted by $PD(\mathcal{R})$, corresponds to the VC-dimension of $\mathcal{R}'$, i.e., $PD(\mathcal{R}) = VCDim(\mathcal{R}')$. 
\end{definition}


Theorem \ref{teo:esample} states that having an upper bound to the pseudo-dimension of a range space allows to build an $\epsilon$-representative sample. 

\begin{theorem}[see \cite{li2001}, Section 1] \label{teo:esample}
    Let $\mathcal{R}' = (D,\mathcal{F}^+)$ be a range space ($D = U \times [0,1]$) with $VCDim(\mathcal{R}') \leq d$ and a probability distribution $\pi$ on $U$. Given $0 < \epsilon,\delta \leq 1$, let $S \subseteq D$ be a collection of elements sampled w.r.t. $\pi$, with 
    \[|S| = \frac{c}{\epsilon^2} \left( d + \ln \frac{1}{\delta} \right)\]
    where $c$ is a universal positive constant. Then $S$ is $\epsilon$-representative with probability at least $1-\delta$.
\end{theorem}

In the work of \cite{loffler2009shape}, it has been proven that the constant $c$ is approximately $\frac{1}{2}$. Lemmas \ref{lemma:atmostone} and \ref{lemma:nozero}, stated an proved by Riondato and Upfal (2018), present constraints on the sets that can be shattered by a range set $\mathcal{F}^+$. 

\begin{lemma}[see \cite{RiondatoUpfal}, Section 3.3] \label{lemma:atmostone}
    Let $B \subseteq D$ be a set that is shattered by $\mathcal{F}^+$. Then, $B$ can contain at most one $(d,y) \in D$ for each $d \in U$ and for a $y \in [0,1]$.
\end{lemma}

\begin{lemma}[see \cite{RiondatoUpfal}, Section 3.3] \label{lemma:nozero}
    Let $B \subseteq D$ be a set that is shattered by $\mathcal{F}^+$. Then, $B$ does not contain any element in the form $(d,0) \in D$, for each $d \in U$.
\end{lemma}

\section{Pseudo-dimension and percolated shortest paths}\label{sec:pseudo}

In this section we model the percolation centrality in terms of a range set of the percolated shortest paths. 
That is, for a given a graph $G = (V,E)$ and the percolation states $x_v$ for each $v \in V$, let $\mathcal{H} = V$, with $n = |V|$, and let $U = S_G$, where $S_G = \bigcup\limits_{(u,w) \in V^2 : u \neq w} S_{uw}$. For each $v \in V$, there is a set $\tau_v = \{p \in U : v \in Int(p)\}$. For a pair $(u,w) \in V^2$ and a path $p_{uw} \in S_G$, let $f_v : U \rightarrow [0,1]$ be the function

\[f_v(p_{uw}) = \frac{R(x_u - x_w)}{\sum\limits_{(f,d) \in V^2 : f \neq v \neq d} R(x_f - x_d)} \mathds{1}_{\tau_v}(p_{uw}).\]

The function $f_v$ gives the proportion of the percolation between $u$ and $w$ to the total percolation in the graph if $v \in Int(p_{uw})$. We define $\mathcal{F} = \{f_v : v \in V\}$.

Let $D = U \times [0,1]$. For each $f_v \in \mathcal{F}$, there is a range $R_v = R_{f_v} = \{(p_{uw},t) : p_{uw} \in U \text{ and } t \leq f_v(p_{uw})\}$. Note that each range $R_v$ contains the pairs $(p_{uw}, t)$, where $0 < t \leq 1$ such that $v \in Int(p_{uw})$ and $t \leq \frac{R(x_u - x_w)}{\sum\limits_{(f,d) \in V^2 : f \neq v \neq d}R(x_f - x_d)}$. We define $\mathcal{F}^+ = \{R_v : f_v \in \mathcal{F}\}$.

Each $p_{uw} \in U$ is sampled according to the function $\pi(p_{uw}) = \frac{1}{n(n-1)}\frac{1}{\sigma_{uw}}$ (which is a valid probability distribution according to Theorem \ref{theo:probdist}), and $\mathbb{E}[f_v(p_{uw})] = p(v)$ for all $v \in V$, as proved in Theorem \ref{theo:expec}.

\begin{theorem}\label{theo:probdist}
    The function $\pi(p_{uw})$, for each $p_{uw} \in U$, is a valid probability distribution.
\end{theorem}

\begin{proof}
    Let $S_{uw}$ be the set of shortest paths from $u$ to $w$, where $u \neq w$. Then,
 \[
\begin{split}
\sum\limits_{p_{uw} \in U} \pi(p_{uw}) 
&= \sum\limits_{p_{uw} \in U} \frac{1}{n(n-1)}\frac{1}{\sigma_{uw}} \\
&= \sum_{u \in V} \sum_{\substack{w \in V \\ w \neq u}} \sum_{p \in S_{uw}} \frac{1}{n(n-1)}\frac{1}{\sigma_{uw}}\\
&= \sum_{u \in V} \sum_{\substack{w \in V \\ w \neq u}} \frac{1}{n(n-1)}\frac{\sigma_{uw}}{\sigma_{uw}}\\
&= \frac{1}{n(n-1)} \sum_{u \in V} \sum\limits_{\substack{w \in V \\ w \neq u}} 1\\
&= \frac{1}{n(n-1)} \sum_{u \in V} (n-1)= 1.
\end{split} 
 \]    
\end{proof}

\begin{theorem} \label{theo:expec}
    For $f_v \in \mathcal{F}$ and for all $p_{uw} \in U$, such that each $p_{uw}$ is sampled according to the probability function $\pi(p_{uw})$, \[\mathbb{E}[f_v(p_{uw})] = p(v).\]
\end{theorem}

\begin{proof}
    For a given graph $G = (V,E)$ and for all $v \in V$, we have from Definition \ref{def:averages}    
    \begin{align*}
        L_U(f_v)&= \mathbb{E}_{p_{uw} \in U}[f_v(p_{uw})]\\ &= \sum_{p_{uw} \in U} \pi(p_{uw}) f_v(p_{uw})\\
        &= \sum_{p_{uw} \in U} \frac{1}{n(n-1)} \frac{1}{\sigma_{uw}} \frac{R(x_u - x_w)}{\sum\limits_{\substack{(f,d) \in V^2 \\ f \neq v \neq d}} R(x_f - x_d)} \mathds{1}_{\tau_v}(p_{uw})\\
        &\scalemath{0.95}{= \frac{1}{n(n-1)} \sum\limits_{\substack{u \in V \\ u \neq v}} \sum\limits_{\substack{w \in V \\ w \neq v \neq u}} \sum_{p \in S_{uw}} \frac{1}{\sigma_{uw}} \frac{R(x_u - x_w)}{\sum\limits_{\substack{(f,d) \in V^2 \\ f \neq v \neq d}} R(x_f - x_d)} \mathds{1}_{\tau_v}(p)}\\
        &= \frac{1}{n(n-1)} \sum\limits_{\substack{u \in V \\ u \neq v}} \sum\limits_{\substack{w \in V \\ w \neq v \neq u}} \frac{\sigma_{uw}(v)}{\sigma_{uw}} \frac{R(x_u - x_w)}{\sum\limits_{\substack{(f,d) \in V^2 \\ f \neq v \neq d}} R(x_f - x_d)}\\
        &= \frac{1}{n(n-1)}\sum\limits_{\substack{(u,w) \in V^2 \\ u \neq v \neq w}} \frac{\sigma_{uw}(v)}{\sigma_{uw}} \frac{R(x_u - x_w)}{\sum\limits_{\substack{(f,d) \in V^2 \\ f \neq v \neq d}} R(x_f - x_d)} = p(v).
    \end{align*}
\end{proof}

Let $S = \{(p_{u_iw_i}, 1 \leq i \leq r)\}$ be a collection of $r$ shortest paths sampled independently and identically from $U$. 
Next, we define $\tilde{p}(v)$, the estimation to be computed by the algorithm, as the empirical average from Definition \ref{def:averages}:

\[\begin{split}
\tilde{p}(v) = L_S(f_v) 
&= \frac{1}{r} \sum\limits_{p_{u_i w_i} \in S} f_v(p_{u_i w_i})\\ 
&= \frac{1}{r} \sum\limits_{p_{u_i w_i} \in S} \frac{R(x_{u_i} - x_{w_i})}{\sum\limits_{\substack{(f,d) \in V^2 \\ f \neq v \neq d}} R(x_f - x_d)} \mathds{1}_{\tau_v}(p_{u_i w_i}).
\end{split}\]

\section{Approximation to the percolation centrality} \label{sec:approximation}

We present an algorithm which its correctness and running time relies on the sample size given by Theorem \ref{teo:esample}. In order to bound the sample size, in Theorem \ref{theo:pseudoperc}, we prove an upper bound to the range space $\mathcal{R}$. We are aware that the main idea in the proof is similar the proof of a result for a different range space on the shortest paths obtained in \cite{Riondato2016} in their work using VC-dimension. For the sake of clarity, instead of trying to fit their definition to our model and use their result, we found it easier stating and proving the theorem directly for our range space.

\begin{theorem} \label{theo:pseudoperc}
Let $\mathcal{R} = (U, \mathcal{F})$ and $\mathcal{R}' = (D, \mathcal{F}^+)$ be the corresponding range spaces for the domain and range sets defined in Section \ref{sec:pseudo}, and let $diam(G)$ be the vertex-diameter of $G$. We have
    \[PD(\mathcal{R}) = VCDim(\mathcal{R}') \leq \lfloor \lg diam(G) - 2 \rfloor + 1.\]
\end{theorem}

\begin{proof}
    Let $VCDim(\mathcal{R}') = k$, where $k \in \mathbb{N}$. Then, there is $S \subseteq D$ such that $|S| = k$ and $S$ is shattered by $\mathcal{F}^+$. From Lemmas \ref{lemma:atmostone} and \ref{lemma:nozero}, we know that for each $p_{uw} \in U$, there is at most one pair $(p_{uw}, t)$ in $S$ for some $t \in [0,1]$ and there is no pair in the form $(p_{uw},0)$. By the definition of shattering, each $(p_{uw}, t) \in S$ must appear in $2^{k-1}$ different ranges in $\mathcal{F}^+$. On the other hand, each pair $(p_{uw},t)$ is in at most $|p_{uw}|-2$ ranges in $\mathcal{F}^+$, since $(p_{uw},t) \notin R_v$ either when $t > f_v(p_{uw})$ or $v \notin Int(p_{uw})$. Considering that $|p_{uw}|-2 \leq diam(G)-2$,  we have
    \[2^{k-1} \leq |p_{uw}|-2 \leq diam(G)-2\]
    \[k-1 \leq \lg (diam(G)-2).\]
    Since $k$ must be integer, $k \leq \lfloor \lg diam(G)-2 \rfloor + 1 \leq \lg (diam(G)-2) + 1$. Finally,
    \[PD(\mathcal{F}) = VCDim(\mathcal{F}^+) = k \leq \lfloor \lg diam(G)-2 \rfloor+1.\]
\end{proof}

By Theorem \ref{theo:expec} and Definition 3, $L_U(f_v) = p(v)$ and $L_S(f_v) = \tilde{p}(v)$, respectively, for each $v \in V$ and $f_v \in \mathcal{F}$. Thus, $|L_S(f_v) - L_U(f_v)| = |\tilde{p}(v) - p(v)|$, and by Theorems \ref{teo:esample} and \ref{theo:pseudoperc}, we have that a sample of size $\lceil\frac{c}{\epsilon^2}\left( \lfloor \lg diam(G) - 2 \rfloor + 1 -\ln \delta \right)\rceil$ suffices to our algorithm, for given $0 < \epsilon,\delta \leq 1$. The problem of computing the diameter of $G$ is not known to be easier than the problem of computing all of its shortest paths \cite{Aingworth:1996:FED:313852.314117}, so obtaining an exact value for the diameter would defeat the whole purpose of using a sampling strategy that avoids computing all shortest paths. Hence, we use an 2-approximation for the diameter described in \cite{Riondato2016}. We note that the diameter can be approximated within smaller factors, but even for a $(\frac{3}{2},\frac{3}{2})$-approximation algorithm (i.e., an algorithm that outputs a solution of size at most $\frac{3}{2} \cdot diam(G) + \frac{3}{2}$) the complexity is $\tilde{\mathcal{O}}(m\sqrt{n} + n^2)$ \cite{Aingworth:1996:FED:313852.314117}, what would also be a bottleneck to our algorithm. Furthermore, since in our case we do not need the largest shortest path, but simply the value of the diameter, and we take logarithm of this value, the approximation of \cite{Riondato2016} is sufficient.

\subsection{Algorithm description and analysis} \label{subsec:algorithm}

    Given a weighted graph $G = (V,E)$ and the percolation states $x_v$ for each $v \in V$ as well as the quality and confidence parameters $0 < \epsilon,\delta \leq 1$, assumed to be constants (they do not depend on the size of $G$) respectively, the Algorithm \ref{alg:percolation} works as follows. At the beginning of the execution the approximated value $diam(G)$ for the vertex-diameter of $G$ is obtained, in line 2, by a 2-approximation described in \cite{Riondato2016}, if the graph is undirected, as previously mentioned. We also compute this approximation in directed graphs ignoring the edges directions, which may not guarantee the approximation factor of 2, but it is good in practice as shown in the experimental evaluation (Section 5). According to Theorem \ref{theo:pseudoperc}, this value is used to determine the sample size, denoted by $r$, in line 3.
    
    The value $minus\_s[v] = \sum\limits_{(f,d) \in V^2 : f \neq v \neq d} R(x_f-x_d)$ for each $v \in V$, which are necessary to compute $\tilde{p}(v)$, is obtained in line 5 by the linear time dynamic programming strategy presented in Algorithm \ref{alg:differences}. The correctness of Algorithm \ref{alg:differences} is not self evident, so we provide a proof of its correctness in Theorem \ref{theo:alg1}. 
    
    A pair $(u,w) \in V^2$ is sampled uniformly and independently, and then a shortest path $p_{uw}$ between $(u,w)$ is sampled uniformly in $S_{uw}$ in lines 10--16. For a vertex $z \in Int(p_{uw})$, the value $\frac{1}{r}\frac{R(x_u-x_w)}{minus\_s[v]}$ is added to $\tilde{p}(z)$.
    
    

\begin{theorem}\label{theo:alg1}
    For an array $A$ of size $n$, sorted in non-decreasing order, Algorithm \ref{alg:differences} returns for $sum$ and $minus\_sum[k]$, respectively, the values $\sum\limits_{i=1}^n \sum\limits_{j = 1}^n  R(A[j]-A[i])$ and $\sum\limits_{\substack{i=1 \\ i \neq k}}^n \sum\limits_{\substack{j = 1 \\ j \neq k}}^n R(A[j]-A[i])$, for each $k \in \{1,\ldots,n\}$.
\end{theorem}

\begin{proof}
    By the definition of $sum$, we have that
    \[\begin{split}
    sum &= \sum_{i=1}^n \sum_{j=1}^n R(A[i]-A[j])\\ 
    &= \sum_{i=1}^n \sum_{j=1}^n R(A[j]-A[i])\\ 
    &= \sum_{i=1}^n \sum_{j=1}^n \max\{A[j]-A[i],0\}.
    \end{split}\]
    
    Since $A$ is sorted, then $\max\{A[j]-A[i],0\} = 0$ if $j<i$. Hence, if we consider only the $j \geq i$, this value becomes
    \[sum = \sum_{i=1}^n \sum_{j=i}^n (A[j]-A[i]).\]
    
    A similar step can be applied to the values of the array $minus\_sum$, and then for all indices $k \in \{1,...,n\}$, 
    \[\scalemath{0.94}{
    minus\_sum[k] = \sum\limits_{\substack{i=1 \\ i \neq k}}^n \sum\limits_{\substack{j=1 \\ j \neq k}}^n \max\{A[j]-A[i],0\} =  \sum\limits_{\substack{i=1 \\ i \neq k}}^n \sum\limits_{\substack{j=i \\ j \neq k}}^n (A[j]-A[i]).
    }\]
    
    The recurrences below follow directly from lines 5 and 6, where $sum_k$ denotes the value of $sum$ at the beginning of the $k$-th iteration of the algorithm. 
    
  \[svp[k]=\begin{cases}
    0, & \text{if $k=1$}\\
     svp[k-1]+ A[k-1], & \text{otherwise}.
  \end{cases}
  \]
    
   \[sum_k = \begin{cases}
    0, & \text{if $k=1$}\\
     sum_{k-1} + (k-1)A[k] - svp[k], & \text{otherwise}.
  \end{cases}
   \]
  
    The solutions to the above recurrences are, respectively,
    \[svp[k] = \sum_{i=1}^{k-1} A[i] \quad \text{ and } \quad sum_k = \sum_{i=1}^{k}\left((i-1)A[i] - svp[i]\right).\]

    The value $sum$ is then correctly computed in lines 4--6, since
	\begingroup
    \allowdisplaybreaks
    \begin{gather*}
        sum = \sum_{i=1}^n \sum_{j=i}^n (A[j]-A[i]) = \sum_{i=1}^n \sum_{j=i}^n A[j] - \sum_{i=1}^n \sum_{j=i}^n A[i]\\
        = \sum_{i=1}^n \sum_{j=i}^n A[j] - \sum_{i=1}^n (n-i+1)A[i]\\
        = \sum_{j=1}^n \sum_{i=1}^j A[j] - \sum_{i=1}^n (n-i+1)A[i]\\
        = \sum_{j=1}^n jA[j] - \sum_{i=1}^n (n-i+1)A[i] = \sum_{i=1}^n iA[i] - \sum_{i=1}^n (n-i+1)A[i]\\
        = \sum_{i=1}^n (i-1)A[i]-\sum_{i=1}^n (n-i)A[i] = \sum_{i=1}^n (i-1)A[i]-\sum_{i=1}^n \sum_{j=1}^{i-1} A[j]\\
        = \sum_{i=1}^n \left( (i-1)A[i]-\sum_{j=1}^{i-1}A[j] \right) = \sum_{i=1}^n \left( (i-1)A[i]-svp[i] \right).
    \end{gather*}
    \endgroup
     
    Finally, $minus\_sum$ is also correctly computed in lines 8 and 9, since
    
    \begingroup
    \allowdisplaybreaks
    \begin{align*}
        & minus\_sum[k] = \sum\limits_{\substack{i=1 \\ i \neq k}}^n \sum\limits_{\substack{j=i \\ j \neq k}}^n (A[j]-A[i])\\
        &\scalemath{0.96}{= \sum_{i=1}^n \sum_{j=i}^n (A[j]-A[i]) - \left( \sum_{j=1}^{k-1}(A[k]-A[j]) + \sum_{j=k+1}^{n}(A[j]-A[k]) \right)}\\
        &= sum - \left( \sum_{j=1}^{k-1} A[k] - \sum_{j=k+1}^n A[k] - \sum_{j=1}^{k-1} A[j] + \sum_{j=k+1}^n A[j] \right)\\
        &\scalemath{0.93}{= sum - \Bigg( (k-1)A[k]-(n-(k+1)+1)A[k]- \sum_{j=1}^{k-1} A[j] + \sum_{j=k+1}^n A[j]\Bigg)}\\
        &\scalemath{0.96}{= sum - \left( (2k-n-1)A[k] + \sum_{j=1}^n A[j] - \sum_{j=1}^{k-1}A[j] - A[k] -\sum_{j=1}^{k-1}A[j] \right)}\\
        &= sum - \left( (2k-n-2)A[k] + \sum_{j=1}^n A[j] - 2\sum_{j=1}^{k-1} A[j] \right)\\
        &= sum - (2k-n-2)A[k] - svp[n+1] + 2svp[k].
    \end{align*}   
	\endgroup

 \end{proof}

\begin{algorithm}[!htbp]
		\SetAlgoNoLine
		\SetAlgoNoEnd
		\DontPrintSemicolon
		\KwData{Array $A$, sorted in non-decreasing order, and $n = |A|$.}
		\KwResult{The value $sum = \sum\limits_{i=1}^n \sum\limits_{j = 1}^n  R(A[j]-A[i])$ and the array $\{minus\_sum[k] = \sum\limits_{\substack{i=1 \\ i \neq k}}^n \sum\limits_{\substack{j = 1 \\ j \neq k}}^n R(A[j]-A[i]), \forall k \in \{1,...,n\}$\}, such that $R(z) = \max\{z,0\}$.}
		
		sum $\leftarrow 0$\;
		
		minus\_sum[i] $\leftarrow 0, \forall i \in \{1,\ldots,n\}$\;
	
	    svp $\leftarrow (0,0, \ldots, 0)$\; 
	    
	    \For{$i \leftarrow 2$ \KwTo $n$}{
			$\textrm{svp}[i] \leftarrow \textrm{svp}[i-1] + A[i-1]$\;
			
			$\textrm{sum} \leftarrow \textrm{sum} + (i-1)A[i] - \textrm{svp}[i]$\;
	    }
		
		$\textrm{svp}[n+1] \leftarrow \textrm{svp}[n] + A[n]$
	
	    \For{$i\leftarrow 1$ \KwTo $n$}{
		    $\textrm{minus\_sum}[i] \leftarrow \textrm{sum} - A[i](2i-n-2) - \textrm{svp}[n+1] + 2\textrm{svp}[i]$\; 
	    }
		
		\textbf{return }{sum, minus\_sum}
		\caption{\textsc{getPercolationDifferences($A,n$)}}
		\label{alg:differences}
	\end{algorithm}
    
    \begin{algorithm}[!htbp]
		\SetAlgoNoLine
		\SetAlgoNoEnd
		\DontPrintSemicolon	
		\KwData{Graph $G = (V,E)$ with $n = |V|$, percolation states $x$, accuracy parameter $0 < \epsilon \leq 1$, confidence parameter $0 < \delta \leq 1$.}
		
		\KwResult{Approximation $\tilde{p}(v)$ for the percolation centrality of all vertices $v \in V$.}
		

	    $\tilde{p}[v] \leftarrow 0, \; minus\_s[v] \leftarrow 0 \text{, } \forall v \in V$\;
		
		$diam(G) \leftarrow \textsc{getVertexDiameter}(G)$\;
		
		$r \leftarrow \lceil\frac{c}{\epsilon^2}\left( \lfloor \lg diam(G) - 2 \rfloor + 1 -\ln \delta \right)\rceil $\;
		
		sort $x$ \tcc*[r]{after sorted, $x = (x_1, x_2, \ldots, x_n)$}
		
		$minus\_s \leftarrow \textsc{getPercolationDifferences}(x,n)$\;
		
		\For{$i \leftarrow 1 $ to $r$}{
		    sample $u \in V$ with probability $1/n$\;
		    
		    sample $w \in V$ with probability $1/(n-1)$\;
	
			$S_{uw} \leftarrow \textsc{allShortestPaths}(u,w)$\;
				
			\If{$S_{uw} \neq \emptyset$}{
				$t \leftarrow w$\\
					
				\While{$t \neq u$}{
					sample $z \in P_u(t)$ with probability $\frac{\sigma_{uz}}{\sigma_{ut}}$\;
					
					\If{$z \neq u$}{
					  $\tilde{p}[z] \leftarrow \tilde{p}[z] + \frac{1}{r} \frac{R(x_u-x_w)}{minus\_s[z]}$\;
					}
					
					$t \leftarrow z$\;
				}
			}
		}
		\textbf{return }{$\tilde{p}[v], \forall v \in V$}
		\caption{\textsc{PercolationCentralityApproximation($G$,$x$,$\epsilon$,$\delta$)}}
		\label{alg:percolation}
	\end{algorithm}

\begin{theorem}\label{theo:alg2}
    Let $S = \{p_{u_i w_i},\ldots,p_{u_r w_r}\}$ be a sample of size $r = \frac{c}{\epsilon^2}(\lfloor \lg diam(G) - 2 \rfloor + 1) - \ln \delta)$ for a given weighted graph $G = (V,E)$ and for given $0 < \epsilon,\delta \leq 1$. Algorithm \ref{alg:percolation} returns with probability at least $1-\delta$ an approximation $\tilde{p}(v)$ to $p(v)$, for each $v \in V$, such that $\tilde{p}(v)$ is within $\epsilon$ error.
\end{theorem}

\begin{proof}
    Each pair $(u_i,w_i)$ is sampled with probability $\frac{1}{n(n-1)}$ in lines 7 and 8, and for each pair, the set $S_{u_i w_i}$ is computed by 
    Dijkstra algorithm (line 9). A shortest path $p_{u_i w_i}$ is sampled independently and uniformly in $S_{u_i w_i}$ (lines 10--16), i.e., with probability $\frac{1}{\sigma_{u_i w_i}}$, by a backward traversing starting from $w_i$ (Lemma 5 in \cite{Riondato2016}, Section 5.1). Therefore, $p_{u_i w_i}$ is sampled with probability $\frac{1}{n(n-1)}\frac{1}{\sigma_{u_i w_i}}$.
    
    In lines 12--16, each $z \in p_{u_i w_i}$ reached by the backward traversing have its value increased by $\frac{1}{r} \frac{R(x_{u_i} - x_{w_i})}{minus\_s[z]}$. The value of $minus\_s[z]$ is correctly computed as shown in Theorem \ref{theo:alg1}. Let $S' \subseteq S$ be the set of shortest paths that $z$ is an internal vertex. Then, at the end of the $r$-th iteration, 
    \[\begin{split}
    \tilde{p}(z)
    &= \frac{1}{r} \sum\limits_{p_{gh} \in S'} \frac{R(x_{g} - x_{h})}{\sum\limits_{\substack{(f,d) \in V^2 \\ f \neq z \neq d}} R(x_f - x_d)}\\ 
    &= \frac{1}{r} \sum\limits_{p_{u_i w_i} \in S} \frac{R(x_{u_i} - x_{w_i})}{\sum\limits_{\substack{(f,d) \in V^2 \\ f \neq z \neq d}} R(x_f - x_d)} \mathds{1}_{\tau_v}(p_{u_i w_i})
    \end{split}\]
    which corresponds to $\tilde{p}(z) = \frac{1}{r}\sum\limits_{\substack{p_{u_i w_i} \in S}} f_z(p_{u_i w_i}).$ Since $L_S(f_v) = \tilde{p}(v)$ and $L_U(f_v) = p(v)$ (Theorem \ref{theo:expec}) for all $v \in V$ and $f_v \in \mathcal{F}$, and $S$ is a sample such that $|S| = r$, then $\Pr(|\tilde{p}(v)-p(v)| \leq \epsilon) \geq 1-\delta$ (Theorem \ref{teo:esample}).
\end{proof}

\begin{theorem}\label{theo:time}
    Given a weighted graph $G=(V,E)$ with $n = |V|$ and $m = |E|$ and a sample of size $r = \frac{c}{\epsilon^2}(\lfloor \lg diam(G) - 2 \rfloor + 1) - \ln \delta)$, Algorithm 2 has running time 
    $\mathcal{O}(r(m \log n))$.
\end{theorem}

\begin{proof}
    We use the linear time algorithm of \cite{vose1991} for the sampling step in lines 7, 8 and 13. The upper bound to the vertex-diameter, computed in line 2 and denoted by $diam(G)$, is obtained by a 2-approximation of \cite{Riondato2016}, which runs in time 
    $\mathcal{O}(m \log n)$. 
    
    
    Sorting the percolation states array $x$ (line 4) can be done in $\mathcal{O}(n \log n)$ time and to the execution of Algorithm 1 on the sorted array $x$ (line 5) has running time $\mathcal{O}(n)$. 
    As for the loop in lines 12--16, the complexity analysis is as follows. Once $|P_u(w)| \leq d_G(w)$, where $d_G(w)$ denotes the degree of $w$ in $G$, and since this loop is executed at most $n$ times if the sampled path traverses all the vertices of $G$, the total running time of these steps corresponds to $\sum_{v \in V} d_G(v) = 2m = \mathcal{O}(m)$. 
    
    The loop in lines 6--16 runs $r$ times and the Dijkstra algorithm which is executed in line 9 has running time $\mathcal{O}(m \log n)$, so the total running time of Algorithm 2 is $\mathcal{O}(n \log n + r\max(m,m \log n)) = \mathcal{O}(n \log n + r(m \log n)) = \mathcal{O}(r(m \log n))$. 
    
\end{proof}

\begin{corollary}\label{cor:a}
    Given an unweighted graph $G=(V,E)$ with $n = |V|$ and $m = |E|$ and a sample of size $r = \frac{c}{\epsilon^2}(\lfloor \lg diam(G) - 2 \rfloor + 1) - \ln \delta)$, Algorithm 2 has running time 
    $\mathcal{O}(r(m + n))$.
\end{corollary}

\begin{proof}
    The proof is analogous to the one of Theorem \ref{theo:time}, with the difference that the shortest paths between a sampled pair $(u,w) \in V^2$ will be computed by the BFS algorithm, which has running time $\mathcal{O}(m+n)$.
\end{proof}

We observe that, even though it is an open problem whether there is a $\mathcal{O}(n^{3-c})$ algorithm for computing all shortest paths in weighted graphs, in the unweighted case there is a $\mathcal{O}(n^{2.38})$ (non-combinatorial) algorithm for this problem \cite{SEIDEL1995400}. However,  even if this algorithm could be adapted to compute betweenness/percolation centrality (what is not clear), our algorithm obtained in Corollary \ref{cor:a} is still faster.

\section{Experimental evaluation}

In Section 5.1 we perform an experimental evaluation of our approach using real world networks. Since, as far as we know, our algorithm is the first estimation algorithm for percolation centrality, we compare our results with the best known algorithm for computing the same measure in the exact case. As expected, our algorithm is many times faster than the exact algorithm, since the exact computation of the percolation centrality takes $\mathcal{O}(n^3)$ time, which is extremely time consuming. Additionally, a main advantage of our algorithm is that it outputs an estimation with a very small error. In fact, for every one of the networks used in our experiment, the average estimation error are kept below the quality parameter $\epsilon$ by many orders of magnitude. For instance, even for $\epsilon = 0.1$ (the largest in our experiments), the average estimation error is in the order of $10^{-11}$ and maximum estimation error of any one of the centrality measure is in the order of $10^{-9}$.

Additionally, we also run experiments using synthetic graphs in order to validate the scalability of our algorithm, since for this task we need a battery of ``similar'' graphs of increasing size. We use power law graphs generated by the Barabasi-Albert model \cite{Barabasi1999} for such experiments. 

We use Python 3.8 language in our implementation and, for graph manipulation, we use the NetworkX library \cite{hagberg2008}. The implementation of the exact algorithm for the percolation centrality is the one available in referred NetworkX library. The experiments were performed on a 2.8 Mhz Intel i7-4900MQ quad core with 12GB of RAM and Windows 10 64-bit operating system.

In all experiments, we set the percolation state $x_v$, for each $v \in V$, as a random number between 0 and 1. The parameters $\delta$ and $c$ remained fixed. We set $\delta = 0.1$ for all experiments. 
and $c = 0.5$ as suggested by Löffler and Phillips (2009) \cite{loffler2009shape}.

\subsection{Real-world graphs}

In this section we describe the experimental results obtained in our work. We use graphs from the {\it Stanford Large Network Dataset Collection} \cite{snapnets}, which are publicly available real world graph datasets. These spans graphs from social, peer-to-peer, autonomous systems and collaboration networks. The details of the datasets can be seen in Table \ref{tab:datasets}. In this table, $\widetilde{VD}(G)$ is the upper bound for the vertex-diameter of the graph obtained in the execution of our algorithm. The last two columns in the same table show by which factor our algorithm is faster than the algorithm for computing the exact centrality. We compute this factor by dividing the running time of the exact method by the running time of our algorithm. We run our algorithm five times and show in the table minimum and the maximum factors. Note that for the largest graph our algorithm is around 30 times faster than the exact method. For this table, the running time of our algorithm is taken for $\epsilon = 0.04$. 

In our experimental evaluation, for every graph, we also run our algorithm for different values of $\epsilon$, more precisely set to 0.04, 0.06, 0.08 and 0.1. In Table \ref{tab:time} we show how many times our algorithm is faster than the exact method in the average of the five runs. 

The error of the estimation computed by our algorithm is within $\epsilon$ for every vertex of every graph even though this guarantee could possibly fail with probability $\delta = 0.1$. However, in addition to the better confidence results than the theoretical guarantee, the most surprising fact is that for every graph used in the experiment the maximum error among the error for the estimation of every vertex is around $10^{-9}$ and average error among all vertices is around $10^{-11}$, even when we set the quality guarantee to $\epsilon = 10^{-1}$. We run our algorithm five times for every graph for each value of $\epsilon$. Them maximum error is taken among all five runs. These results are shown in Figures \ref{fig:errors-avg}, \ref{fig:times} and \ref{fig:errors-max}.


\begin{table*}[!htbp]
\begin{tabular}{ccccccc}
\toprule
\multicolumn{1}{l}{} & \multicolumn{1}{l}{} & \multicolumn{1}{l}{} & \multicolumn{1}{l}{} & \multicolumn{1}{l}{}       & \multicolumn{2}{c}{\textbf{$\frac{\text{exact time}}{\text{approx. time}}$}} \\
\hline
\textbf{Graph}       & \textbf{Type}        & \textbf{$|V|$}       & \textbf{$|E|$}       & \textbf{$\widetilde{VD}(G)$} & \textbf{Min}                  & \multicolumn{1}{l}{\textbf{Max}}                 \\
\midrule
wiki-Vote            & Directed             & 7115                 & 103689     & 11     &     10.21                         & 10.49                                             \\
p2p-Gnutella08         & Directed             & 6301                 & 20777                & 14                             & 10.23                         & 9.5                                               \\
oregon1-010331       & Undirected           & 10670                & 22002                &   14                         & 17.72                         & 17.6                                              \\
ca-CondMat           & Undirected           & 23133                & 93497                &        21                    & 31.19                         & 29.44                                          \\ 
asas20000102          & Undirected, weighted         & 6474                & 13895                &        14                    & 9.18                         & 9.72                                        \\ 
\bottomrule                                       
\end{tabular}
\caption{\label{tab:datasets}Datasets details for the real-world graphs.}
\end{table*}

\begin{table}[!htbp]
\begin{tabular}{ccccc}
\hline
\multicolumn{5}{c}{\textbf{exact/approx. time average ratio}}                 \\
\hline
\textbf{Graph} & \textbf{$\epsilon=0.04$} & \textbf{$\epsilon=0.06$} & \textbf{$\epsilon=0.08$} & \textbf{$\epsilon=0.1$} \\
\hline
wiki-Vote      & 10.36         & 22.17         & 36.93         & 52.71        \\
p2p-Gnutella   & 9.92          & 19.59         & 33.91         & 53.21        \\
oregon1-010331 & 17.66         & 39.7          & 70.14         & 108.71       \\
ca-CondMat     & 30.25         & 67.32         & 120.98        & 191.1        \\ 
as20000102     & 9.61         & 20.67         & 38.81        & 55.4        \\ 
\bottomrule
\end{tabular}
\caption{\label{tab:time} Ratio for average running time after 5 runs of the percolation centrality estimation algorithm and the exact algorithm.}
\end{table}


\begin{figure}[!htbp]
\begin{tikzpicture}[baseline]
\begin{axis}[
    title={Undirected graphs (avg, avg. + std.)},
    xlabel={Parameter $\epsilon$},
    ylabel={Absolute estimated error},
    xtick={0.04,0.06,0.08,0.1},
    legend pos=north west,
    legend style={font=\fontsize{8}{2}\selectfont},
    ymajorgrids=true,
    grid style=dashed,
    legend entries={oregon-1 (avg), ca-CondMat (avg), as20000102 (avg), oregon-1 (avg+std), ca-CondMat (avg+std), as20000102 (avg+std)},
    scaled x ticks = false,
    x tick label style={rotate=45,anchor=east,/pgf/number format/fixed},   
]
 
\addplot[color=blue, mark=square]
    coordinates {(0.04,0.00000000000101638005)(0.06,0.00000000000143485533)(0.08,0.00000000000171465186)(0.1,0.00000000000185326010)};

\addplot[color=red, mark= triangle]
    coordinates {(0.04,0.00000000000030559745)(0.06,0.00000000000038643671)(0.08,0.00000000000043599531)(0.1,0.00000000000048109812)};
    
\addplot[color=purple, mark=square]
    coordinates {(0.04,0.00000000000454905824)(0.06,0.00000000000590264034)(0.08,0.00000000000745477663)(0.1,0.00000000000813738920)};
    
\addplot[color=blue, mark= x]
    coordinates {(0.04,0.00000000000814867520)(0.06,0.00000000001234523001)(0.08,0.00000000001670834819)(0.1,0.00000000002004271066)};
    
\addplot[color=red, mark= *]
    coordinates {(0.04,0.00000000000157162092)(0.06,0.00000000000219787193)(0.08,0.00000000000284095346)(0.1,0.00000000000359368372)};
    
\addplot[color=green, mark=triangle]
    coordinates {(0.04,0.00000000003212459968)(0.06,0.00000000004690706174)(0.08,0.00000000006146947796)(0.1,0.00000000007444463698)};

\end{axis}
\end{tikzpicture}
\hskip 0pt 
\begin{tikzpicture}[baseline]
\begin{axis}[
    title={Directed graphs (avg, avg. + std.)},
    xlabel={Parameter $\epsilon$},
    ylabel={Absolute estimated error},
    xtick={0.04,0.06,0.08,0.1},
    legend pos=north west,
    legend style={font=\fontsize{8}{2}\selectfont},
    ymajorgrids=true,
    grid style=dashed,
    legend entries={p2p-Gnutella08 (avg), wiki-Vote (avg), p2p-Gnutella08 (avg+std), wiki-Vote (avg+std)},
    scaled x ticks = false,
    x tick label style={rotate=45,anchor=east,/pgf/number format/fixed},   
]
 
\addplot[color=blue, mark=square]
    coordinates {(0.04,0.00000000000722717993)(0.06,0.00000000000912665191)(0.08,0.00000000001044590006)(0.1,0.00000000001193180382)};

\addplot[color=red, mark= triangle]
    coordinates {(0.04,0.00000000000030559745)(0.06,0.00000000000038643671)(0.08,0.00000000000043599531)(0.1,0.00000000000048109812)};
    
\addplot[color=blue, mark= x]
    coordinates {(0.04,0.00000000003390526320)(0.06,0.00000000004600142698)(0.08,0.00000000005759188486)(0.1,0.00000000007335331008)};
    
\addplot[color=red, mark= *]
    coordinates {(0.04,0.00000000001139021297)(0.06,0.00000000001573088543)(0.08,0.00000000002090518039)(0.1,0.00000000002230077515)};

\end{axis}
\end{tikzpicture}
\caption{Percolation centrality absolute error estimation.}
\label{fig:errors-avg}
\end{figure}
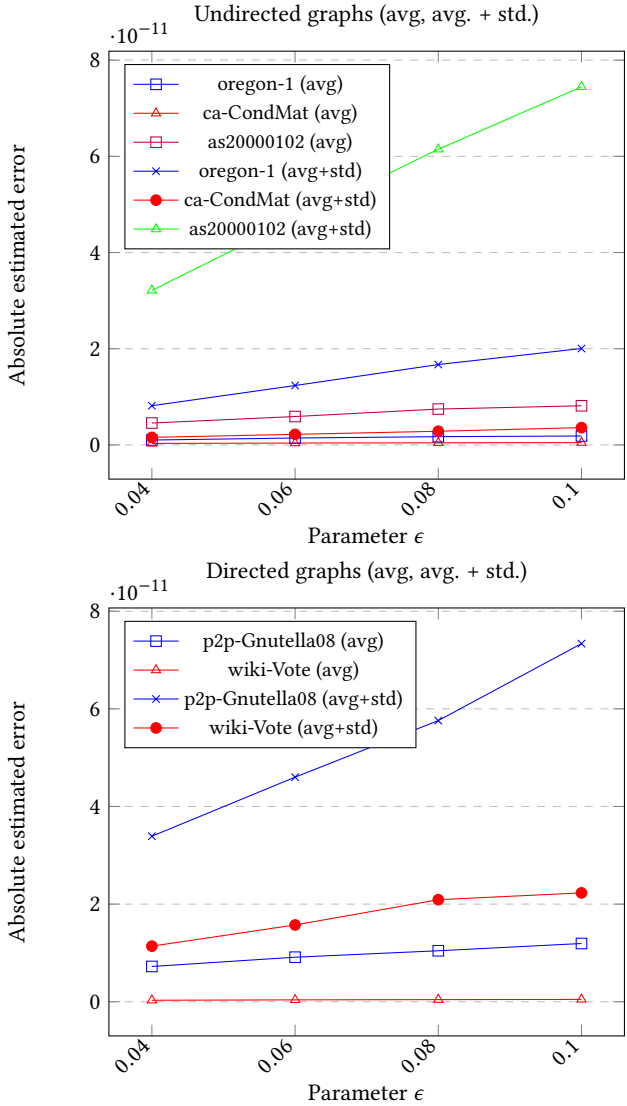

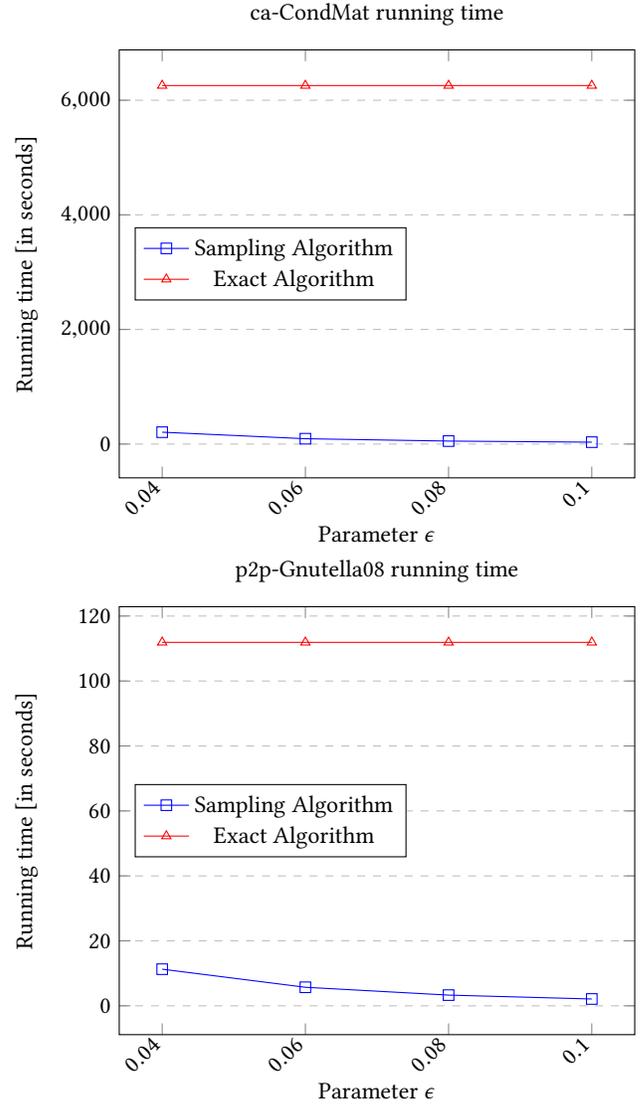
\begin{figure}[!htbp]
\begin{tikzpicture}[baseline]
\begin{axis}[
    title={ca-CondMat running time},
    xlabel={Parameter $\epsilon$},
    ylabel={Running time [in seconds]},
    xtick={0.04,0.06,0.08,0.10},
    legend style={at={(0.03,0.5)},anchor=west},
    ymajorgrids=true,
    grid style=dashed,
    legend entries={Sampling Algorithm, Exact Algorithm},
    scaled x ticks = false,
    x tick label style={rotate=45,anchor=east,/pgf/number format/fixed},   
]
 
\addplot[color=blue, mark=square]
    coordinates {(0.04,206.8007974088001)(0.06,92.93150820160008)(0.08,51.71349199519991)(0.1,32.75314521399978)};

\addplot[color=red, mark=triangle]
    coordinates {(0.04,6256.0497520809995)(0.06,6256.0497520809995)(0.08,6256.0497520809995)(0.1,6256.0497520809995)};

\end{axis}
\end{tikzpicture}
\hskip 0pt 
\begin{tikzpicture}[baseline]
\begin{axis}[
    title={p2p-Gnutella08 running time},
    xlabel={Parameter $\epsilon$},
    ylabel={Running time [in seconds]},
    xtick={0.04,0.06,0.08,0.1},
    legend style={at={(0.03,0.5)},anchor=west},
    ymajorgrids=true,
    grid style=dashed,
    legend entries={Sampling Algorithm, Exact Algorithm},
    scaled x ticks = false,
    x tick label style={rotate=45,anchor=east,/pgf/number format/fixed},   
]
 
\addplot[color=blue, mark=square]
    coordinates {(0.04,11.2748659226)(0.06,5.7107732584000015)(0.08,3.299517869399989)(0.1,2.102881264199999)};

\addplot[color=red, mark=triangle]
    coordinates {(0.04,111.894270283)(0.06,111.894270283)(0.08,111.894270283)(0.1,111.894270283)};
    
\end{axis}
\end{tikzpicture}

\caption{Average running time for 5 runs of Algorithm \ref{alg:percolation}}
\label{fig:times}
\end{figure}

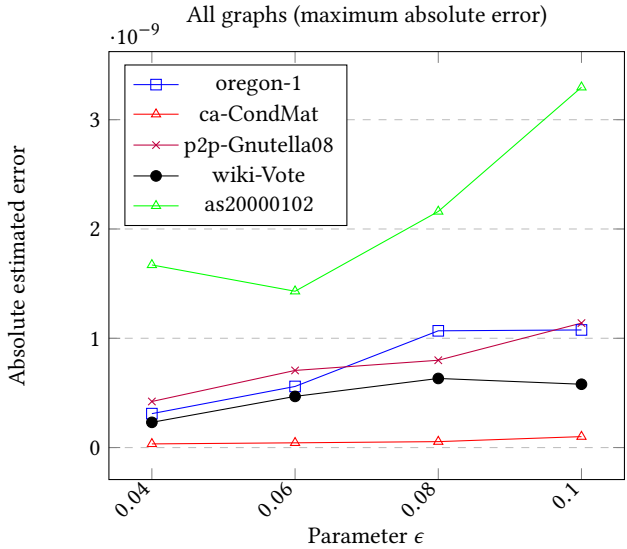
\begin{figure}[!htbp]
\begin{tikzpicture}[baseline]
\begin{axis}[
    title={All graphs (maximum absolute error)},
    xlabel={Parameter $\epsilon$},
    ylabel={Absolute estimated error},
    xtick={0.04,0.06,0.08,0.1},
    legend pos=north west,
    ymajorgrids=true,
    grid style=dashed,
    legend entries={oregon-1, ca-CondMat, p2p-Gnutella08, wiki-Vote, as20000102},
    scaled x ticks = false,
    x tick label style={rotate=45,anchor=east,/pgf/number format/fixed},   
]
 
\addplot[color=blue, mark=square]
    coordinates {(0.04,0.00000000031040146402)(0.06,0.00000000055921546725)(0.08,0.00000000106830154488)(0.1,0.00000000107617891932)};

\addplot[color=red, mark= triangle]
    coordinates {(0.04,0.00000000003357480632)(0.06,0.00000000004343663394)(0.08,0.00000000005435797734)(0.1,0.00000000010017460702)};
    
\addplot[color=purple, mark= x]
    coordinates {(0.04,0.00000000042160563055)(0.06,0.00000000070572903943)(0.08,0.00000000079899723150)(0.1,0.00000000113890078961)};
    
\addplot[color=black, mark= *]
    coordinates {(0.04,0.00000000023084031958)(0.06,0.00000000046823092958)(0.08,0.00000000063238826804)(0.1,0.00000000057950267240)};
    
\addplot[color=green, mark=triangle]
    coordinates {(0.04,0.00000000167112844219)(0.06,0.00000000143029174139)(0.08,0.00000000216030877219)(0.1,0.00000000329659545586)};

\end{axis}
\end{tikzpicture}

\caption{Percolation centrality absolute maximum error estimation.}
\label{fig:errors-max}
\end{figure}


\subsection{Synthetic graphs}

In the scalability experiments, we used a sequence of synthetic graphs increasing in size and compared the execution time of our estimating algorithm with the algorithm for the computation of the exact centrality provided by NetworkX library. 

We also use the same library for generating random power law graphs, by the Barabasi-Albert model \cite{Barabasi1999} with each vertex creating two edges, obtaining power law graphs with average degree of $2$. 
In the experiments we use graphs with the number of vertices $n$ in $\{100, 250, 500, 1000, 2000, 4000, 8000, 16000\}$. The value of $\epsilon$ is fixed at $0.05$. The results are shown in Figure \ref{fig:scalability-random}.


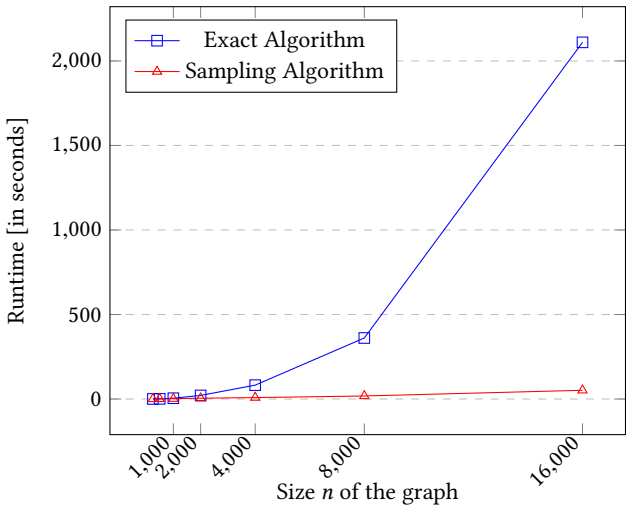
\begin{figure}[!htbp]
\begin{tikzpicture}[baseline]
\begin{axis}[
    xlabel={Size $n$ of the graph},
    ylabel={Runtime [in seconds]},
    xtick={1000,2000, 4000, 8000, 16000},
    legend pos=north west,
    ymajorgrids=true,
    grid style=dashed,
    legend entries={Exact Algorithm, Sampling Algorithm},
    scaled x ticks = false,
    x tick label style={rotate=45,anchor=east,/pgf/number format/fixed},   
]
 
\addplot[color=blue, mark=square]
    coordinates {(250,0.253)(500,1.112)(1000,4.753)(2000,21.193)(4000,82.0781)(8000,361.003)(16000,2110.084)};
 
 \addplot[color=red, mark= triangle]
    coordinates {(250,0.478)(500,1.084)(1000,2.003)(2000,4.643)(4000,8.453)(8000,17.943)(16000,51.637)};
\end{axis}
\end{tikzpicture}
%
%
%
%
\caption{Scalability experiments, with $\epsilon = 0.05$.}
\label{fig:scalability-random}
\end{figure}

\section{Conclusion}


We presented an algorithm with running time $\mathcal{O}(m \log^2 n)$ for estimating the percolation centrality for every vertex of a weighted graph. The estimation obtained by our algorithm is within $\epsilon$ of the exact value with probability $1- \delta$, for {\it fixed} constants $0 < \epsilon,\delta \leq 1$. The running time of the algorithm is reduced to $\mathcal{O}((m+n)\log n)$ if the input graph is unweighted. Since many large scale graphs are sparse and have small diameter (typically of size $\log n$), our algorithm provides a fast approximation for such graphs (more precisely running in $\mathcal{O}(n \log n \log \log n$) time). We validate our proposed method performing experiments on real world networks. As expected, our algorithm is much faster than the exact algorithm (around 30 times for the largest graph) and, furthermore, in practice the estimation error is many orders of magnitude smaller than the theoretical worst case guarantee for graphs of a variety of sizes.


\begin{acks}
This work was partially funded by CNPq(Proc. 428941/2016-8) and CAPES.
\end{acks}

\bibliographystyle{ACM-Reference-Format}
\bibliography{percolation}

\end{document}